\documentclass[conference]{IEEEtran}
\usepackage{amsfonts}
\usepackage[margin=0.75in]{geometry}
\usepackage{graphicx}
\usepackage{color}
\usepackage{amsmath,amsfonts,amssymb,amsthm,epsfig,epstopdf,url,array}
\usepackage{url,textcomp}
\usepackage{authblk}
\usepackage{cite}
\makeatletter
\newcommand*{\rom}[1]{\expandafter\@slowromancap\romannumeral #1@}
\makeatother
\newcommand{\bs}{\boldsymbol}
\newtheorem{theorem}{Theorem}
\newtheorem{lemma}{Lemma}
\newtheorem{corollary}{Corollary}

\newtheorem{remark}{Remark}

\begin{document}
\title{A General Analytical Approach for Outage Analysis of HARQ-IR over Correlated Fading Channels}
\author[$\ddag$]{Zheng~Shi}
\author[$\dag$]{Shaodan~Ma}
\author[$\ddag$]{Guanghua~Yang}
\author[$\dag$]{Kam Weng Tam}
\author[$\S$]{Ming-Hua~Xia}
\affil[$\ddag$]{The School
of Electrical and Information Engineering and Institute of Physical Internet, Jinan University, China}
\affil[$\dag$]{Department of Electrical and Computer Engineering, University of Macau, Macau}
\affil[$\S$]{School of Electronics and Information Technology, Sun Yat-Sen University, Guangzhou, China.}
% with Institute of Physical Internet, Jinan University (ghyang@jnu.edu.cn)
%\author{Zheng~Shi,
%%        Haichuan~Ding,
%        Shaodan~Ma,
%        Guanghua~Yang,
%        Kam-Weng~Tam,
%        and Minghua Xia
%%\thanks{Manuscript received January 14, 2015; revised May 29, 2015; accepted July 19, 2015. The associate editor coordinating the review of this paper and approving it for publication was M. Elkashlan.}
%%\thanks{Zheng Shi, Shaodan Ma and Kam-Weng Tam are with the Department of Electrical and Computer Engineering, University of Macau, Macao (e-mail:shizheng0124@gmail.com, shaodanma@umac.mo, kentam@umac.mo).}
%%\thanks{Haichuan Ding was with University of Macau, and is now with the Department of Electrical and Computer Engineering, University of Florida, U.S.A. (email: dhcbit@gmail.com).}
%%\thanks{Su Pan is with Nanjing University of Posts and Telecommunications, China (email: supan@njupt.edu.cn).}
%%\thanks{The corresponding author is Shaodan Ma.}
%%\thanks{This work was supported by the Research Committee of University of Macau under grants: MYRG078(Y1-L2)-FST12-MSD and MYRG101(Y1-L3)-FST13-MSD.}
%}
%\affil{Department of Electrical and Computer Engineering, University of Macau, Macau}

\maketitle
\begin{abstract}
This paper proposes a general analytical approach to derive the outage probability of hybrid automatic repeat request with incremental redundancy (HARQ-IR) over correlated fading channels in closed-form. Unlike prior analyses, the consideration of channel correlation is one of the key reasons making the outage analysis involved. Using conditional Mellin transform, the outage probability of HARQ-IR over \textit{correlated Rayleigh} fading channels is exactly expressed as a mixture of outage probabilities of HARQ-IR over \textit{independent Nakagami} fading channels, where the weights are negative multinomial probabilities. Its straightforward application is to conduct asymptotic outage analysis to gain more insights, in which the asymptotic outage probability is obtained in a concise form. The asymptotic outage probability possesses some special properties which ease optimal power allocation and rate selection of HARQ-IR. Finally, numerical results are presented for validations and discussions.

%%It is revealed that low time correlation is beneficial to the outage performance and full time diversity can be achieved even under time-correlated fading channels.

%Due to the infinite series representation of outage probability, an efficient truncation approach is proposed to enable its computation. The analytical result, of independent interest, facilitates the asymptotic outage analysis. Meaningful insights can be extracted from the asymptotic outage probability, which clearly shows the impacts of time correlation, transmission rate and transmit powers. It is also revealed that low time correlation is beneficial to the outage performance, and full time diversity can be achieved even under time-correlated fading channels. Remarkably, these analytical results can serve as a solid foundation for optimal system design and performance evaluation. Finally, numerical results are presented for validations and discussions.
\end{abstract}
\IEEEpeerreviewmaketitle

\section{Introduction}\label{sec:int}
In the last decade we have witnessed an ever increasing research interest in hybrid automatic repeat request (HARQ) owing to its high potential for reliable reception. It has been rigorously proved in \cite{caire2001throughput} that HARQ with incremental redundancy (HARQ-IR) can reach the ergodic capacity in Gaussian collision channels. Moreover, HARQ-IR outperforms other types of HARQ (i.e., Type I HARQ and HARQ with chase combining) since extra coding gain is achieved. %Since HARQ-IR accumulates mutual information from each HARQ round, the accumulated mutual information is then given by a sum of mutual information.

%The analysis of HARQ-IR is of great importance for system design and has been widely investigated in the literature. As shown in \cite{makki2014performance}, the analysis of HARQ-IR is essentially to determine cumulative distribution function (CDF) of accumulated mutual information. In fact, the CDF of the accumulated mutual information can be directly applied to obtain the most fundamental performance metric of HARQ, i.e., outage probability.

Although HARQ-IR has been studied extensively in the literature, most of them assume either quasi-static \cite{shen2009average,makki2013green} or independent fading channels \cite{to2015power,jinho2013energy,yilmaz2009productshifted,
chelli2013performance,larsson2014throughput}. For example, \cite{shen2009average} and \cite{makki2013green} conduct average rate performance analysis and power optimization for HARQ-IR, respectively, under quasi-static fading channels. The quasi-static fading channels assume the same channel realization experienced in all HARQ rounds. However, the analytical results under quasi-static fading channels only apply to low mobility environment. In contrast, the transmitted signals in high mobility environment usually undergo independent fading channels, where the channels are assumed to vary independently from one transmission to another. In this scenario, the analysis becomes to determine the cumulative distribution function (CDF) of the product of multiple shifted independent random variables (RVs). Several approaches have been developed to handle this problem in the literature. For example, in \cite{to2015power}, Log-normal approximation is proposed based on central limit theorem (CLT). Nevertheless, this approximation can not ease the optimal system design because of the complex results. Moreover, \cite{jinho2013energy} adopts Jensen's inequality to obtain the lower bound of the CDF to reduce the computational complexity. In order to obtain the exact CDF, two methods are developed, i.e., Mellin transform \cite{yilmaz2009productshifted,chelli2013performance} and multi-fold convolution \cite{larsson2014throughput}. Unfortunately, the analytical results obtained from the two methods are too complicated to apply to optimal design.

In addition, another widely adopted channel model is correlated fading channel, which usually occurs in low-to-medium mobility environment. %Under time-correlated fading channels, the corresponding CDF of accumulated mutual information becomes more challenging because of involving a product of multiple shifted correlated RVs. It is also totally different from the analysis of HARQ-CC over time-correlated fading channels in \cite{kim2011optimal,jin2011optimal}, where a sum of multiple correlated RVs is concerned.
To our best knowledge, only few approximation approaches have been proposed to analyze the performance of HARQ-IR over time-correlated fading channels, i.e., Log-normal approximation \cite{yang2014performance}, polynomial fitting technique \cite{shi2015analysis}, inverse moment matching method \cite{shi2016inverse} and Jensen's inequality \cite{shi2016optimal}. Unfortunately, Log-normal approximation in \cite{yang2014performance} is inaccurate under fading channels with medium-to-high correlation. To improve the approximation accuracy, other two approximation approaches are developed in \cite{shi2015analysis,shi2016inverse}. However, the proposed approximation approaches admit a mean square error (MSE) convergence, which does not necessarily imply that the approximation error would approach to zero at every point \cite[p86]{adams2013continuous}. Particularly, the proposed approximation approaches become unstable when evaluating a low outage probability, which limits their application to conducting the asymptotic outage analysis to further extract insightful results. Besides, the analytical results in \cite{shi2015analysis,shi2016inverse} are still too complicated to apply to practical system design. To avoid that, a tractable approximate expression is derived for outage probability based on Jensen's inequality, and is then applied to obtain optimal transmission powers in closed-form in \cite{shi2016optimal}. However, there exists an evident gap between the approximate and the exact optimal solutions.

%\cite{makki2014performance} demonstrates that the essential parameter to evaluate HARQ-IR is the accumulated mutual information, with whose cumulative distribution function (CDF) can be directly applied to obtaining the most fundamental performance metric of HARQ, i.e., outage probability. It is proved that the derivation of the CDF of accumulated mutual information is equivalent to determining the CDF of a product of multiple shifted random variables (RVs). Specifically, for Rayleigh fading channels, these RVs comply with exponential distribution.

In this paper, a general analytical approach is proposed to derive the outage probability of HARQ-IR over exponentially correlated Rayleigh fading channels. By using conditional Mellin transform, the outage probability is exactly derived as a representation of a mixture of outage probabilities of HARQ-IR over independent Nakagami fading channels where the weights are negative multinomial probabilities. The analytical result, of independent interest, facilitates asymptotic outage analysis. Specifically, the asymptotic outage probability is derived in a compact form which clearly quantifies the impacts of channel time correlation, transmission rate and powers. %It is revealed that the time correlation of fading channels has a detrimental effect on system performance, while full time diversity can be achieved by HARQ-IR even under time-correlated fading channels.
Moreover, the special properties of asymptotic outage probability facilitate the optimal system design.

The rest of this paper is structured as follows. Section \ref{sec:sys_mod} introduces the system model and formulates the outage probability of HARQ-IR. Exact and asymptotic outage analyses are conducted in Section \ref{sec:exc} and \ref{sec:asy}, respectively. Section \ref{sec:num_res} verifies our numerical results. Finally, some conclusions are drawn in Section \ref{sec:con}.

\section{System Model and Outage Formulation}\label{sec:sys_mod}
Due to space limitation, this paper considers a point-to-point HARQ-IR enabled system operating over exponentially correlated Rayleigh fading channels. It is worth emphasizing that the general analytical approach proposed in this paper can be easily applied to more complicated HARQ-IR system. To facilitate outage formulation, HARQ-IR protocol and time-correlated fading channels are first introduced.
\subsection{HARQ-IR protocol}\label{sec:harq_ir}
Following HARQ-IR, each original information message is first encoded into $K$ sub-codewords at the source, where $K$ denotes the maximum allowable number of transmissions for each message. These $K$ sub-codewords will be delivered one by one until the message is successfully decoded. In each transmission, all the erroneously received sub-codewords are combined with currently received sub-codewords for joint decoding. If success, an acknowledgement (ACK) message is fed back to the source and the transmission for the next information message will be initiated. Otherwise, a negative acknowledgement (NACK) message is fed back to notify the source, and the next sub-codeword is transmitted until the maximum number of transmissions $K$ is reached. %Here we assume an error-free feedback channel, that is, all feedback signals can be correctly decoded.
\subsection{Time-correlated Rayleigh fading channels}\label{sec:corr}
Denote ${\bf x}_k$ as the $k$th sub-codeword of length $L$. We assume a block fading channel, %which means that each symbol of ${\bf x}_k$ experiences an identical channel realization during the $k$th transmission. Accrodingly,
the signal received in the $k$th transmission is thus given by
\begin{equation}\label{eqn:sign_mod}
  {\bf y}_k = h_k {\bf x}_k + {\bf n}_k,
\end{equation}
where ${\bf n}_k$ denotes a complex additive white Gaussian noise (AWGN) vector with mean vector ${\bf 0}_L$ and covariance matrix ${\bf I}_{L}$, i.e., ${\bf n}_k \sim {\cal CN}(0,{{\bf I}_{L}})$, ${\bf I}_{L}$ represents an identity matrix, $h_k$ denotes the block Rayleigh fading channel coefficient in the $k$th transmission%, i.e., the magnitude of $h_k$ follows a Rayleigh distribution
. In this paper, time correlation of fading channels is considered. A commonly adopted correlated Rayleigh fading channel model is given as \cite{kim2011optimal}
\begin{equation}\label{eqn:R_k_def}
{h_k} =  {{{{\sigma _k}}}} \left( {\sqrt {1 - {\rho ^{2\left( {k + \delta  - 1} \right)}}} {\alpha _k} + {\rho ^{k + \delta  - 1}}{\alpha _0}} \right),
\end{equation}
where $\rho$, $\delta$ and ${\sigma_k}^2$ denote the time correlation, the channel feedback delay and the mean squared magnitude of ${h_k}$, respectively; ${\alpha_{1}},\cdots,{\alpha_{K}}$ and ${\alpha_{0}}$ are independent circularly-symmetric complex normal RVs with zero mean and unit variance, i.e., ${\alpha_{0}}, {\alpha_{k}} \sim \mathcal{CN}\left( {0,1} \right)$.
%
% Add a footnote to mention the correlation matrix of fading channels.
%
%Specifically, the correlation coefficient between the squared channel amplitudes $|h_k|^2$ and $|h_l|^2$ as
%\begin{align}
%\label{eqn_cor_fa}
%{\rho _{l,k}} &= \frac{{{\rm{E}}\left( {{{\left| {{h_{l}}} \right|}^2}{{\left| {{h_{k}}} \right|}^2}} \right) - {\rm{E}}\left( {{{\left| {{h_{l}}} \right|}^2}} \right){\rm{E}}\left( {{{\left| {{h_{k}}} \right|}^2}} \right)}}{{\sqrt {{\rm{Var}}\left( {{{\left| {{h_{l}}} \right|}^2}} \right){\rm{Var}}\left( {{{\left| {{h_{k}}} \right|}^2}} \right)} }} = {\lambda _{l}}^2{\lambda _{k}}^2,\, 1 \le l \ne k \le K,
%\end{align}
%where ${\rm Var}(\cdot)$ denotes the operation of variance.

%It is worth noting that the correlated channel model generalizes quasi-static fading channels and fast fading channels as its special cases. More specifically, $|{\boldsymbol{\lambda }}| = {\bf 1}_K$ indicates that the channel magnitudes are fully correlated, namely, quasi-static fading channels, where ${\bf 1}_K$ denotes an all-ones vector of length $K$.
%While for the case of $|{\boldsymbol{\lambda }}| = {\bf 0}_K$, the channel magnitudes are independent and Nakagami-m distributed, namely, fast fading channels, where ${\bf 0}_K$ is a null vector of length $K$. Without a loss of generality, index $K$ is omitted in the remainder.

According to (\ref{eqn:sign_mod}), the received signal-to-noise ratio (SNR) in the $k$th transmission is% given by
\begin{equation}\label{eqn:SNR}
{\gamma _k} = {P_k}{\left| {{h_k}} \right|^2}{\rm{ = }}{P_k}{\sigma _k}^2{\left| {\sqrt {1 - {\rho ^{2\left( {k + \delta  - 1} \right)}}} {\alpha _k} + {\rho ^{k + \delta  - 1}}{\alpha _0}} \right|^2},
\end{equation}
where $P_k$ is the transmitted signal power in the $k$th transmission. Since the channel magnitude $\left| {{h_k}} \right|$ is Rayleigh distributed, ${\gamma _k} $ complies with exponential distribution with mean $P_k{\sigma_k}^2$. Due to the correlation between fading channels, the SNRs $\bs \gamma = (\gamma _1,\gamma _2,\cdots,\gamma _K)$ are thus correlated exponential RVs.
%To proceed with our analysis, the joint probability density function (PDF) of correlated SNRs $\bs \gamma$ will be first derived.

\subsection{Outage Formulation}
Outage probability is proved as the most fundamental performance metric of HARQ schemes \cite{caire2001throughput}. %For HARQ-IR,  the outage probability is directly determined by the CDF of accumulated mutual information.
Specifically, assuming that information-theoretic capacity achieving channel coding is adopted for HARQ-IR, an outage event happens when the accumulated mutual information $I_K$ is below the transmission rate $\mathcal R$. The outage probability after $K$ transmissions is thus written as
\begin{equation}\label{eqn:out_prob_def}
{p_{out,K}} = \Pr \left( {{I_K}  < \mathcal R} \right) = F_{{I_K}}(\mathcal R),
\end{equation}
where $F_{{I_K}}(\cdot)$ denotes the CDF of ${I_K}$. Following the HARQ-IR protocol, the accumulated mutual information after $K$ transmissions is given by
\begin{equation}\label{eqn:def:accum_inf}
{I_K} = \sum\limits_{k = 1}^K {{{\log }_2}\left( {1 + {\gamma _k}} \right)}.
\end{equation}
Accordingly, the outage probability ${p_{out,K}}$ becomes as
\begin{equation}\label{eqn:out_prob_rew}
{p_{out,K}} = \Pr \left( {G \triangleq \prod\limits_{k = 1}^K {\left( {1 + {\gamma _k}} \right)}  < {2^{\mathcal R}}} \right) = {F_G}\left( {{2^{\mathcal R}}} \right),
\end{equation}
where $F_G(\cdot)$ denotes the CDF of $G$. Therefore, noting that $\bs \gamma$ are correlated exponential RVs, the derivation of outage probability ${p_{out,K}}$ essentially turns to determining the CDF of the product of multiple shifted correlated-exponential RVs, i.e., ${F_G}\left( {{x}} \right)$.

%As aforementioned, the essential parameter to characterize the performance of HARQ-IR is the CDF of accumulated mutual information. Following HARQ-IR protocol, the accumulated mutual information after $K$ transmissions is given by
%\begin{equation}\label{eqn:def:accum_inf}
%{I_K} = \sum\limits_{k = 1}^K {{{\log }_2}\left( {1 + {\gamma _k}} \right)}.
%\end{equation}
%
%As the most fundamental performance metric of HARQ schemes, the outage probability of HARQ-IR is directly determined by the CDF of accumulated mutual information \cite{caire2001throughput}. Specifically, assuming that information-theoretic capacity achieving channel coding is adopted for HARQ-IR, an outage event happens when the accumulated mutual information is below the transmission rate $\mathcal R$. The outage probability after $K$ transmissions is thus given by
%\begin{equation}\label{eqn:out_prob_def}
%{p_{out,K}} = \Pr \left( {{I_K}  < \mathcal R} \right) = F_{{I_K}}(\mathcal R).
%\end{equation}
%where $F_{{I_K}}(\cdot)$ denotes the CDF of ${I_K}$. Hereby, to derive ${p_{out,K}}$, it amounts to determining
%\begin{equation}\label{eqn:out_prob_rew}
%{p_{out,K}} = \Pr \left( {G \triangleq \prod\limits_{k = 1}^K {\left( {1 + {\gamma _k}} \right)}  < {2^{\mathcal R}}} \right) = {F_G}\left( {{2^{\mathcal R}}} \right),
%\end{equation}
%where $G$ represents the product of multiple shifted correlated-exponential RVs and $F_G(\cdot)$ denotes the CDF of $G$. %Then the derivation of outage probability essentially turns to determine the CDF of $G$.
\section{Exact Outage Analysis}\label{sec:exc}
%\subsection{The CDF of $G$}
Unlike the outage analysis under independent fading channels \cite{to2015power,jinho2013energy,yilmaz2009productshifted,
chelli2013performance,larsson2014throughput}, it is more challenging to derive an exact expression for outage probability ${p_{out,K}}$ when fading channels are time-correlated. Recall that the derivation of ${p_{out,K}}$ essentially turns to determining the CDF of the product of multiple shifted correlated-exponential RVs, i.e., $F_G(x)$. To the best of our knowledge, there is no readily available method to solve the problem. Nevertheless, it is found from (\ref{eqn:SNR}) that correlation among the SNRs can be eliminated given the RV ${\alpha_{0}}$, i.e., the SNRs are independent when ${\alpha_{0}}$ is given. This conditional independence of SNRs inspires us to derive the CDF of $F_G(x)$ using conditional Mellin transform. Specifically, it can be proved by using \cite[Theorem 1.3.4]{muirhead2009aspects} that $\bs \gamma$ follows independent noncentral chi-squared distributions with $2$ degrees-of-freedom conditioned on $T = {{{\left| {\alpha_{0}} \right|}^2}} $. The conditional PDF of the SNR $\gamma_k$ given $T$ is given by
\begin{multline}\label{eqn:cond_pdf_r_k_rew}
{{f_{\left. {{\gamma _k}} \right|T}}\left( {\left. {{x_k}} \right|T} \right)} = \frac{1}{\Omega_k}{\exp{\left( - \frac{{{x_k} + {P_k}{\sigma _k}^2{\rho ^{2\left( {k + \delta  - 1} \right)}}T}}{\Omega_k}\right)}} \times\\
{}_0{F_1}\left( {;1;\frac{{{\rho ^{2\left( {k + \delta  - 1} \right)}}T{x_k}}}{{\left( {1 - {\rho ^{2\left( {k + \delta  - 1} \right)}}} \right)\Omega_k}}} \right),\, |\rho| \ne 1,
\end{multline}
where ${}_0F_1(\cdot)$ denotes the confluent hypergeometric limit function \cite[Eq.9.14.1]{gradshteyn1965table}, and $ \Omega_k={{P_k}{\sigma _k}^2{{\left( {1 - {\rho ^{2\left( {k + \delta  - 1} \right)}}} \right)}}}$.

Noticing that Mellin transform is an efficient way to obtain the distribution of a product of multiple independent RVs, the conditional Mellin transform is proposed  to derive the conditional probability density function (PDF) of $G$ given $T$ as follows.

When given $T$, the Mellin transform with respect to the conditional PDF of $G$, ${f_{\left. G \right|T}}(x|t)$, can be written as
\begin{multline}\label{eqn_mellin_Rk1}
\left\{ {\mathcal M{f_{\left. G \right|T}}} \right\}\left( s \right) = {\rm E}\left\{ {\left. {{G^{s - 1}}} \right|T=t} \right\} \triangleq \phi \left( {\left. s \right|t} \right) \\
= \sum\limits_{{l_1}, \cdots ,{l_K} = 0}^\infty  {{t^{\sum\limits_{k = 1}^K {{l_k}} }}{e^{ - t\sum\limits_{k = 1}^K {\frac{{{\rho^{2(k+\delta-1)}}}}{{1 - {\rho^{2(k+\delta-1)}}}}} }}\prod\limits_{k = 1}^K {\frac{{{{\left( {\frac{{{\rho^{2(k+\delta-1)}}}}{{1 - {\rho^{2(k+\delta-1)}}}}} \right)}^{{l_k}}}}}{{{l_k}!{{ {\Omega_k} }^{1 + {l_k}}}}}} }\\
\times \Psi \left( {1 + {l_k},1 + {l_k} + s;\frac{1}{{\Omega_k}}} \right).
\end{multline}
where $\Psi \left( \cdot \right)$ denotes Tricomi's confluent hypergeometric function \cite[9.211,4]{gradshteyn1965table}.

Accordingly, by using inverse Mellin transform, $f_{\left. G \right|T}(x|t)$ can be written as \cite{debnath2010integral}
\begin{equation}\label{eqn_pdf_G_on_T}
{f_{\left. G \right|T}}\left( {\left. x \right|t} \right) = \left\{ {{\mathcal M^{ - 1}}\phi } \right\}\left( s \right) = \frac{1}{{2\pi \rm i}}\int\nolimits_{c - {\rm i}\infty }^{c + {\rm i}\infty } {{x^{ - s}}} \phi \left( {\left. s \right|t} \right)ds.
\end{equation}
where ${\rm i} = \sqrt{-1}$. Putting (\ref{eqn_mellin_Rk1}) into (\ref{eqn_pdf_G_on_T}), it yields
\begin{multline}\label{eqn:g_hat_cond_t}
{f_{\left. G \right|T}}\left( {\left. x \right|t} \right) = \sum\limits_{{l_1}, \cdots ,{l_K} = 0}^\infty  {{t^{\sum\limits_{k = 1}^K {{l_k}} }}{e^{ - t\sum\limits_{k = 1}^K {\frac{{{\rho ^{2(k + \delta  - 1)}}}}{{1 - {\rho ^{2(k + \delta  - 1)}}}}} }}}  \\
\times \prod\limits_{k = 1}^K {\frac{{{{\left( {\frac{{{\rho ^{2(k + \delta  - 1)}}}}{{1 - {\rho ^{2(k + \delta  - 1)}}}}} \right)}^{{l_k}}}}}{{{l_k}!}}} {f_{{{\cal A}_{\bf{l}}}}}(x),
\end{multline}
where ${\bf l} = (l_1,l_2,\cdots,l_K)$ and ${{f_{{{ {\cal A}}_{\bf{l}}}}}(x)}$ denotes the PDF of the product of $K$ independent shifted-Gamma RVs, i.e., $\mathcal A_{\bf l} = \prod\nolimits_{k = 1}^K (1+R_{{\bf l},k})$, and $R_{{\bf l},k} \sim \mathcal {G}(1+l_k,\Omega_k)$. ${{f_{{{ {\cal A}}_{\bf{l}}}}}(x)}$ is explicitly given by (\ref{eqn:pdf_produc_gamma_ind_sec}), as shown at the top of next page,
\begin{figure*}[!t]
%\normalsize
\begin{align}\label{eqn:pdf_produc_gamma_ind_sec}
{f_{{{ {\mathcal A}}_{\bf{l}}}}}(x) %= \frac{1}{{2\pi {\rm{i}}}}\int_{c - {\rm{i}}\infty }^{c + {\rm{i}}\infty } {\prod\limits_{k = 1}^K {\frac{{\Psi \left( {m + {l_k},m + {l_k} + s;\frac{1}{{{\Omega _k}}}} \right)}}{{{{\left( {{\Omega _k}} \right)}^{m + {l_k}}}}}} {x^{ - s}}ds} \\
= \frac{1}{{\prod\limits_{k = 1}^K {\Omega_k} }}Y_{0,K}^{K,0}\left[ {\left. {\begin{array}{*{20}{c}}
 - \\
{\left( {0,1,\frac{1}{{{\Omega _1}}},1 + {l_1}} \right), \cdots ,\left( {0,1,\frac{1}{{{\Omega _K}}},1 + {l_K}} \right)}
\end{array}} \right|\frac{x}{{\prod\limits_{k = 1}^K {\Omega_k} }}} \right].
\end{align}
%\hrulefill
%\vspace*{4pt}
\end{figure*}
where $Y_{p,q}^{m,n}[\cdot]$ defines generalized Fox's H function \cite{yilmaz2009productshifted,chelli2013performance}.

With (\ref{eqn:g_hat_cond_t}), the PDF of $G$ can then be obtained by averaging the conditional probability ${f_{\left. G \right|T}}\left( {\left. x \right|t} \right)$ over the distribution of $T$, such that
\begin{equation}\label{eqn:f_G_PDF}
{f_G}\left( x \right){\rm{ = }}\int_0^\infty  {{f_{\left. G \right|T}}\left( {\left. x \right|t} \right){f_T}\left( t \right)dt},
\end{equation}
where ${f_T}\left( t \right)$ denotes the PDF of $T$. Since ${\alpha_{0}}$ is a complex Gaussian RV with zero mean and unit variance, $T = {{{\left| {\alpha_{0}} \right|}^2}} $ obeys a exponential distribution with the PDF of
\begin{equation}\label{eqn_pdf_T_nak}
{f_T}\left( t \right) = {e^{ - t}},\, t \ge 0.
\end{equation}

Substituting (\ref{eqn:g_hat_cond_t}) and (\ref{eqn_pdf_T_nak}) into (\ref{eqn:f_G_PDF}), the following theorem can be obtained after some mathematical manipulations.
\begin{theorem}\label{the:cdf_pdf_corr_gam_shif}
The PDF of $ G$ can be expressed as a mixture of PDFs of products of independent shifted-Gamma RVs, such that
\begin{equation}\label{eqn:f_G_def_sec}
{f_G}\left( x \right) = {{\rm{E}}_{\bf{l}}}\left\{ {{f_{{\mathcal A_{\bf{l}}}}}(x)}\right\} = \sum\limits_{{{\bf l}} \in {{\mathbb N}_0}^K}  {{W_{\bf{l}}}{f_{{\mathcal A_{\bf{l}}}}}(x)} ,
\end{equation}
where ${{\mathbb N}_0}^K$ denotes $K$-ary Cartesian power of natural number set; ${\bf l}$ follows a negative multinomial distribution whose probability mass function (pmf) is
\begin{multline}\label{eqn:W_l_def_sec}
{W_{\bf{l}}} \triangleq \Pr \left( {{\bf{l}} = \left( {{l_1},{l_2}, \cdots ,{l_K}} \right)} \right)  = \\
{\left( {1 + \sum\limits_{k = 1}^K {\frac{{{\rho ^{2\left( {k + \delta  - 1} \right)}}}}{{1 - {\rho ^{2\left( {k + \delta  - 1} \right)}}}}} } \right)^{ - 1}}{\left({\sum\limits_{k = 1}^K {{l_k}} }\right)!}\prod\limits_{k = 1}^K {\frac{{{w_k}^{{l_k}}}}{{{l_k}!}}}   ,
\end{multline}
with ${w_k} = \frac{{{\rho ^{2\left( {k + \delta  - 1} \right)}}}}{{1 - {\rho ^{2\left( {k + \delta  - 1} \right)}}}}{\left( {1 + \sum\nolimits_{\iota  = 1}^K {\frac{{{\rho ^{2\left( {\iota  + \delta  - 1} \right)}}}}{{1 - {\rho ^{2\left( {\iota  + \delta  - 1} \right)}}}}} } \right)^{ - 1}}$, ${\bf w}=(w_1,\cdots,w_K)$, and $\sum\nolimits_{{\bf l} \in {{\mathbb N}_0}^K}  {{W_{\bf{l}}}}  = 1$, i.e., ${\bf l} \sim {\rm NM}(1,\bf w)$.
\end{theorem}
\begin{proof}
The detailed derivation is omitted due to space limitation.
\end{proof}

Accordingly, the CDF of $G$ can be derived based on (\ref{eqn:f_G_def_sec}), as shown in the following corollary.
\begin{corollary}\label{the:cdf_pdf_corr_gam_shif}
The CDF of $ G$ is given by
\begin{multline}\label{eqn:CDF_G_def_sec}
{F_{ G}}\left( x \right) = \int\nolimits_0^x {{f_{ G}}\left( t \right)dt}  = {{\rm{E}}_{\bf{l}}}\left\{ {\int\nolimits_0^x {{f_{{{ {\mathcal A}}_{\bf{l}}}}}(t)dt}}\right\}\\
={{\rm{E}}_{\bf{l}}}\left\{ {{F_{{\mathcal A_{\bf{l}}}}}(x)}\right\} = \sum\limits_{{{\bf l}} \in {{\mathbb N}_0}^K}  {{W_{\bf{l}}}{F_{{{ {\cal A}}_{\bf{l}}}}}\left( x \right)},
\end{multline}
where ${{F_{{{ {\mathcal A}}_{\bf{l}}}}}(x)}$ denotes the CDF of a product of independent shifted-Gamma RVs $\mathcal {G}(1+l_k,\Omega_k)$, and can be expressed in terms of Fox's H function as (\ref{eqn:CDF_F_A_def_sec}), as shown at the top of next page.
\begin{figure*}[!t]
%\normalsize
\begin{align}\label{eqn:CDF_F_A_def_sec}
{F_{{{ {\mathcal A}}_{\bf{l}}}}}(x) = \int\nolimits_0^x {{f_{{{ {\mathcal A}}_{\bf{l}}}}}(t)dt}
%= Y_{1,K + 1}^{K,1}\left[ {\left. {\begin{array}{*{20}{c}}
%{\left( {1,1,0,1} \right)}  \\
%{\left( {1,1,\frac{1}{{{\Omega _1}}},1 + {l_1}} \right), \cdots ,\left( {1,1,\frac{1}{{{\Omega _K}}},1 + {l_K}} \right),\left( {0,1,0,1} \right)}
%\end{array}} \right|\frac{x}{{\prod\limits_{k = 1}^K {\Omega_k} }}} \right].
 = Y_{K + 1,1}^{1,K}\left[ {\left. {\begin{array}{*{20}{c}}
{\left( {0,1,\frac{1}{{{\Omega _1}}},1 + {l _1}} \right), \cdots ,\left( {0,1,\frac{1}{{{\Omega _K}}},1 + {l _K}} \right),\left( {1,1,0,1} \right)}\\
{\left( {0,1,0,1} \right)}
\end{array}} \right|\frac{{\prod\limits_{k = 1}^K {{\Omega _k}} }}{x}} \right].
\end{align}
\hrulefill
%\vspace*{4pt}
\end{figure*}
%and (\ref{eqn:CDF_F_A_def_sec}) holds by using the property 3 of generalized Fox's H Function in Appendix \ref{app:fox_H_fun}.
\end{corollary}
%\subsection{Exact Outage Probability ${p_{out,K}}$}
It follows from (\ref{eqn:out_prob_rew}) and (\ref{eqn:CDF_G_def_sec}) that the outage probability ${p_{out,K}}$ can be written as
\begin{equation}\label{eqn:out_prob_def_hat}
{p_{out,K}} = {F_{ G}}\left( {{2^{\mathcal R}} } \right)  = \sum\limits_{{{\bf l}} \in {{\mathbb N}_0}^K}  {{W_{\bf{l}}}{F_{{{ {\cal A}}_{\bf{l}}}}}\left( 2^{\mathcal R} \right)}, \, |\rho| \ne 1.
\end{equation}
%Due to the representation of infinite series of (\ref{eqn:out_prob_def_hat}), it is impossible to accurately compute (\ref{eqn:out_prob_def_hat}).
Furthermore, meaningful results can be extracted from (\ref{eqn:out_prob_def_hat}), as given in the following remark.
\begin{remark}
\label{remark}
In fact, ${F_{{{ {\cal A}}_{\bf{l}}}}}\left( 2^{\mathcal R} \right)$ can be regarded as an outage probability of HARQ-IR after $K$ transmissions over independent Nakagami-m fading channels, where ${R_{{\bf{l}},k}} \triangleq {P_k}{\left| {{h_{{{\cal A}_{\bf{l}}},k}}} \right|^2}$ and ${{h_{{{\mathcal A}_{\bf l}},k}}}$ denotes the corresponding channel coefficient in the $k$th transmission, i.e.,
\begin{equation}\label{eqn:channel_coeff_dec}
\left| {{h_{{{\mathcal A}_{\bf l}},k}}} \right| %\triangleq \sqrt {\frac{{{R_{{\bf l},k}}}}{{{P_k}}}} \\
\sim {\rm Nakagami}\left( {1 + {l_k},{\left( {1 + {l_k}} \right)(1-{\rho^{2(k+\delta-1)}}){\sigma_k}^2}} \right).
\end{equation}
Therefore, the outage probability ${p_{out,K}}$ can be treated as a mixture of outage probabilities of HARQ-IR over independent Nakagami fading channels where the weights are probabilities of ${\rm NM}(1,\bf w)$.
\end{remark}
The result in Remark 1 is very useful in the later asymptotic outage analysis.

\section{Asymptotic Outage Analysis}\label{sec:asy}
\subsection{Asymptotic Outage Probability}
Although outage probability of HARQ-IR over time-correlated fading channels can be derived in closed-form in (\ref{eqn:out_prob_def_hat}), straightforward insights still can not be found. Therefore, the asymptotic analysis of outage probability is resorted to extract meaningful insights for HARQ-IR. To proceed, we assume that
\begin{equation}\label{eqn:gamma_snr_rel}
({P_1},{P_2}, \cdots ,{P_K}) = \gamma {\boldsymbol{\theta }},
\end{equation}
where ${\boldsymbol{\theta }} = \left( {{\theta _1},{\theta _2}, \cdots ,{\theta _K}} \right)$ denotes a constant vector associated with power allocation. Noticing that ${p_{out,K}}$ is a weighted sum of ${F_{{{ {\cal A}}_{\bf{l}}}}}\left( 2^{\mathcal R} \right)$, some asymptotic properties with respect to ${F_{{{ {\cal A}}_{\bf{l}}}}}\left( 2^{\mathcal R} \right)$ are proved first under high SNR regime, i.e., $\gamma \to \infty$, to facilitate the asymptotic analysis of ${p_{out,K}}$.
%\subsection{Asymptotic Property of ${F_{{{ {\mathcal A}}_{\bf{l}}}}}(x)$}
Due to the independence of RVs $\{R_{{\bf l},k}\}_{k=1}^K$, the asymptotic expression of ${F_{{{ {\mathcal A}}_{\bf{l}}}}}(x)$ as $\gamma \to \infty$ can be derived.
Specifically, from (\ref{eqn:CDF_F_A_def_sec}), ${F_{{\mathcal A_{\bf{l}}}}}(x)$ can be written in the form of Mellin-Barnes integral, such that
%Putting (\ref{eqn:hat_c_def}) into (\ref{eqn:der_A_1_CDF}), the asymptotic expression for ${F_{{{ {\mathcal A}}_{\bf{l}}}}}(x)$ as $\gamma \to \infty$ can be written as
\begin{equation}\label{eqn:der_A_1_asyCDF}
{F_{{\mathcal A_{\bf{l}}}}}(x) = \frac{1}{{2\pi {\rm{i}}}}\int_{c - {\rm{i}}\infty }^{c + {\rm{i}}\infty } {\frac{{\Gamma \left( s \right)}}{{\Gamma \left( {s + 1} \right)}}\prod\limits_{k = 1}^K {\frac{{\Psi \left( {1,2{\rm{ + }}{l_k} - s;\frac{1}{{{\gamma\zeta _k}}}} \right)}}{{{{\left( {{\gamma\zeta _k}} \right)}^{m + {l_k}}}}}} {x^s}ds},
\end{equation}
where ${\zeta _k} = {\theta _k}{\sigma _k}^2\left( {1 - {\rho ^{2\left( {k + \delta  - 1} \right)}}} \right)$ and $c \in (0,\infty)$.
By adopting \cite[Eq.9.210.2]{gradshteyn1965table}, it produces
\begin{align}\label{eqn:F_hat_A_0_der_fur_ex}
&{F_{{\mathcal A_{\bf{l}}}}}(x) = \frac{1}{{2\pi {\rm{i}}}}\int_{{c_1} - {\rm{i}}\infty }^{{c_1} + {\rm{i}}\infty } {\frac{{\Gamma \left( s \right)}}{{\Gamma \left( {s + 1} \right)}}\prod\limits_{k = 1}^K {\frac{1}{{{{\left( {{\gamma\zeta _k}} \right)}^{1 + {l_k}}}}}} }\times \notag\\
&  \left( {\begin{array}{*{20}{l}}
{{{\frac{{\Gamma \left( { - 1 - {l_k} + s} \right)}}{{\Gamma \left( s \right)}}}_1}{F_1}\left( {1 + {l_k},2 + {l_k} - s;\frac{1}{{{\gamma\zeta _k}}}} \right) + } \\
{\frac{{\Gamma \left( {1 + {l_k} - s} \right)}}{{\Gamma \left( {1 + {l_k}} \right)}}{{\left( {\frac{1}{{{\gamma\zeta _k}}}} \right)}^{s - m - {l_k}}}{}_1{F_1}\left( {s, - {l_k} + s;\frac{1}{{{\gamma\zeta _k}}}} \right)}
\end{array}} \right){x^s}ds.
\end{align}
Since ${}_1{F_1}\left( {\alpha ,\beta ;\frac{1}{{{\gamma\zeta _k}}}} \right)$ can be expanded as
\begin{equation}\label{eqn:F_11_series_exp}
{}_1{F_1}\left( {\alpha ,\beta ;\frac{1}{{{\gamma\zeta _k}}}} \right) = 1 + \frac{\beta }{\alpha }\frac{1}{{{\gamma\zeta _k}}} + o\left( {\frac{1}{{{\gamma}}}} \right),
\end{equation}
where $o(\cdot)$ refers to the little-O notation, and we say $f(\gamma) \in o(\phi(\gamma))$ provided that $\lim \limits_{\gamma \to \infty}f(\gamma)/\phi(\gamma) = 0$.
Clearly, as $\gamma$ approaches to infinity, the dominant term in (\ref{eqn:F_11_series_exp}) is $1$.
 %\begin{align}\label{eqn:F_Al_as_0}
%{F_{{A_{\bs \ell} }}}\left( x \right) &= \Pr \left( {\prod\limits_{k = 1}^K {(1 + {R_{{\bs \ell} ,k}})}  \le x} \right) \notag\\
% &\le \Pr \left( {{R_{{\bs \ell},1}} \le x - 1} \right) = \frac{{\Upsilon \left( {m + {l_1},{\Omega _1}x} \right)}}{{\Gamma \left( {m + {l_1}} \right)}} = O\left( {{x^\infty }} \right)
% \end{align}

Herein, we assume $c>1+{\rm max}\{\bf{l}\}$ because $c_1 $ could be any point in $(0,\infty)$. Therefore, (\ref{eqn:F_hat_A_0_der_fur_ex}) can be further written as
\begin{align}\label{eqn:F_A_0_asym_F_hy_exp}
{F_{{\mathcal A_{\bf{l}}}}}(x) &= \prod\limits_{k = 1}^K {\frac{1}{{{{\left( {\gamma {\zeta _k}} \right)}^{1 + {l_k}}}}}} \frac{1}{{2\pi {\rm{i}}}} \times \notag\\
&\quad \int_{{c} - {\rm{i}}\infty }^{{c} + {\rm{i}}\infty } {\frac{{\Gamma \left( s \right)}}{{\Gamma \left( {s + 1} \right)}}\prod\limits_{k = 1}^K {\frac{{\Gamma \left( { - 1 - {l_k} + s} \right)}}{{\Gamma \left( s \right)}}} {x^s}ds} \notag\\
 &\quad + o\left( {{\gamma ^{ -{d_{{\mathcal A_{\bf{l}}}}}}}} \right)
% &= \prod\limits_{k = 1}^K {\frac{1}{{{{\left( {\gamma {\zeta _k}} \right)}^{1 + {l_k}}}}}} G_{K + 1,K + 1}^{0,K + 1}\left( {\left. {\begin{array}{*{20}{c}}
%{1,2 + {l_1}, \cdots ,2 + {l_K}}\\
%{1, \cdots ,1,0}
%\end{array}} \right|x} \right) \notag\\
%&+ o\left( {{\gamma ^{ - K - \sum\limits_{k = 1}^K {{l_k}} }}} \right).
\end{align}
where ${d_{{A_{\bf{l}}}}} = K + \sum\nolimits_{k = 1}^K {{l_k}} $.
After some manipulations, ${{F_{{{ {\cal A}}_{\bf{l}}}}}(x)}$ can be finally expressed as
\begin{multline}\label{eqn:F_A_0_asym_F_hy_exp1}
{{F_{{{ {\cal A}}_{\bf{l}}}}}(x)} = \prod\limits_{k = 1}^K {{{\left( {\frac{1}{{{\theta _k}{\sigma _k}^2\left( {1 - {\left( {1 - {\rho ^{2\left( {k + \delta  - 1} \right)}}} \right)}^2} \right)}}} \right)}^{1 +l_k }}} \\
G_{K + 1,K + 1}^{0,K + 1}\left( {\left. {\begin{array}{*{20}{c}}
{1,2 +l_1, \cdots ,2 +l_K}\\
{1, \cdots ,1,0}
\end{array}} \right|x} \right){\gamma ^{-d_{{ {\cal A}_{\bf{l}}}}}}  \\
+ o\left( {{\gamma ^{-d_{{ {\cal A}_{\bf{l}}}}}}} \right).
\end{multline}
%It is worth noting that (\ref{eqn:g_0_0_der_meij_rem}) holds for ${{\rm min}\{\bf{l}\}}+m > 0$  considering from Remark \ref{rem:chann}
%that $R_{{\bf{l}},k} \sim \mathcal {G}(m+l_k,{\Omega _k})$, where $l_k + m > 0$.

According to (\ref{eqn:F_A_0_asym_F_hy_exp1}), it is readily found that the following lemma holds.
\begin{lemma}\label{the:cdf_pdf_relatoin_hat}
As $\gamma \to \infty$, the ratio of ${{F_{{ {\mathcal A}_{\bf{l}}}}}\left( x \right)}$ to ${{F_{{ {\mathcal A}_{\bf{0}}}}}\left( x \right)}$ satisfies
\begin{equation}\label{eqn:F_AI_F_A0_rel_hat}
\frac{{{F_{{ {\cal A}_{\bf{l}}}}}\left( x \right)}}{{{F_{{ {\cal A}_{\bf{0}}}}}\left( x \right)}} = o\left( {{1}} \right),\, {{{\bf l}} \in {{\mathbb N}_0}^K} \,{\rm and }\, {\bf l \ne 0}.
\end{equation}
%\begin{equation}\label{eqn:F_AI_F_A0_rel_hat}
%\frac{{{F_{{ {\cal A}_{{\bf{l}}^{\left(1\right)}}}}}\left( x \right)}}{{{F_{{ {\cal A}_{{\bf{l}}^{\left(2\right)}}}}}\left( x \right)}} = o\left( 1 \right),\quad \sum\limits_{k = 1}^K {l_k^{\left( 1 \right)}}  < \sum\limits_{k = 1}^K {l_k^{\left( 2 \right)}} .
%\end{equation}
%where ${{\bf{l}}^{\left( 1 \right)}} = \left( {l_1^{\left( 1 \right)},\cdots,l_K^{\left( 1 \right)}} \right)$ and ${{\bf{l}}^{\left( 2 \right)}} = \left( {l_1^{\left( 2 \right)},\cdots,l_K^{\left( 2 \right)}} \right)$. Hence, for a sufficiently large SNR, i.e., $\gamma \gg 1$, we have
%\begin{equation}\label{eqn:bound_F_A_l}
%\mathop {\max }\limits_{{l_1} + \cdots  + {l_K} \ge N+1} \left( {{F_{{ {\mathcal A}_{\bf{l}}}}}\left( x \right)} \right) = \mathop {\max }\limits_{{l_1} +  \cdots  + {l_K} = N+1} \left( {{F_{{ {\mathcal A}_{\bf{l}}}}}\left( x \right)} \right).
%\end{equation}
\end{lemma}
This lemma will enable the derivation of the asymptotic outage probability ${p_{out\_asy,K}}$.
%\subsection{Asymptotic Analysis of ${p_{out,K}}$}
Specifically, with Lemma \ref{the:cdf_pdf_relatoin_hat}, the CDF ${F_{ G}}\left( x \right)$ can be written as
\begin{align}\label{eqn:cdf_F_shif_gam_pro_asy}
{F_{ G}}\left( x \right) %&= {W_{\bf{0}}}{F_{{{ {\mathcal A}}_{\bf{0}}}}}\left( x \right) + \sum\limits_{{l_1} +  \cdots  + {l_K} > 0}  {{W_{\bf{l}}}{F_{{{ {\mathcal A}}_{\bf{l}}}}}\left( x \right)} \notag \\
&= {W_{\bf{0}}}{F_{{{ {\mathcal A}}_{\bf{0}}}}}\left( x \right)\left( {1 + \frac{1}{{{W_{\bf{0}}}}}\sum\limits_{{l_1} + \cdots  + {l_K} > 0}  {{W_{\bf{l}}}\frac{{{F_{{{ {\mathcal A}}_{\bf{l}}}}}\left( x \right)}}{{{F_{{{ {\mathcal A}}_{\bf{0}}}}}\left( x \right)}}} } \right)\notag\\
& = {W_{\bf{0}}}{F_{{{ {\mathcal A}}_{\bf{0}}}}}\left( x \right)\left( {1 + \frac{1}{{{W_{\bf{0}}}}}\sum\limits_{{l_1} +  \cdots  + {l_K} > 0}  {{W_{\bf{l}}}o\left( 1 \right)} } \right)  \notag \\
&= {W_{\bf{0}}}{F_{{{ {\mathcal A}}_{\bf{0}}}}}\left( x \right)\left( {1 + o\left( 1 \right)} \right),
\end{align}
where the last step holds since $W_{\bf l}$ is irrelevant to $\gamma$ and $\sum\nolimits_{{l_1} + \cdots  + {l_K} > 0}  {{W_{\bf{l}}}} = 1-W_{\bf 0} < 1$.
%From (\ref{eqn:cdf_F_shif_gam_pro_asy}), as $\gamma \to \infty$, the outage probability ${ {p_{out,K}}}$ can be asymptotically written as
%\begin{equation}\label{eqn:asym_out}
%{ {p_{out,K,\infty}}} = {F_{ G}}\left( {{2^{\cal R}}} \right) \approx {W_{\bf{0}}}{F_{{ {\mathcal A}_{\bf{0}}}}}\left( {{2^{\cal R}} } \right), \, |\rho| \ne 1,
%\end{equation}
%which means that the outage probability under time-correlated Rayleigh fading channels can be asymptotically expressed as a weighted outage probability under independent Rayleigh fading channels with mean squared magnitude ${(1-{\rho^{2(k+\delta-1)}}){\sigma_k}^2}$.

With (\ref{eqn:F_A_0_asym_F_hy_exp1}), (\ref{eqn:cdf_F_shif_gam_pro_asy}) can be further simplified as
\begin{align}\label{eqn:der_hat_G_CDF_asy_fur}
{F_{ G}}\left( x \right)% &= {W_{\bf{0}}} c\left( {x,\gamma ,\rho,\delta,{\bf{l}},K,\bs \theta } \right){\gamma ^{ - K}}\left( {1 + o\left( 1 \right)} \right) \notag\\
% &= {W_{\bf{0}}}{\gamma ^{ - mK}}\left( {g_{{\bf{0}},{\bf{0}}}}\left( x \right){\prod\limits_{k = 1}^K {{{\left( {\frac{m}{{{\theta _k}{\sigma _k}^2\left( {1 - {\rho^{2(k+\delta-1)}}} \right)}}} \right)}^m}}  + o\left( 1 \right)} \right)\left( {1 + o\left( 1 \right)} \right) \notag\\
% &= {W_{\bf{0}}}{\gamma ^{ - mK}}{g_{{\bf{0}},{\bf{0}}}}\left( x \right) \prod\limits_{k = 1}^K {{{\left( {\frac{m}{{{\theta _k}{\sigma _k}^2\left( {1 - {\rho^{2(k+\delta-1)}}} \right)}}} \right)}^m}}   \left( {1 + o\left( 1 \right)} \right) \notag \\
 &= {\left( {\ell \left( {\rho,K} \right)\prod\limits_{k = 1}^K {{{{\theta _k}{\sigma _k}^2}}} } \right)^{ - 1}} {\mathcal G_K}\left( x \right) {\gamma ^{ - K}} + o\left( \gamma^{-K} \right),
\end{align}
where $\ell \left( {\rho,K} \right)$ is a function of $\rho$ and $K$, such that
 \begin{equation}\label{eqn:ell_def}
\ell \left( {\rho,K} \right) = \left( {1 + \sum\limits_{k = 1}^K {\frac{{{\rho^{2(k+\delta-1)}}}}{{1 - {\rho^{2(k+\delta-1)}}}}} } \right)\prod\limits_{k = 1}^K {\left( {1 - {\rho^{2(k+\delta-1)}}} \right)}.
\end{equation}
and
\begin{align}\label{eqn:def_calG}
{{\cal G}_K}\left( x \right) &= G_{K + 1,K + 1}^{0,K + 1}\left( {\left. {\begin{array}{*{20}{c}}
{1,2, \cdots ,2}\\
{1, \cdots ,1,0}
\end{array}} \right|x} \right)\notag\\
%&= \frac{1}{{2\pi {\rm{i}}}}\int\nolimits_{c - {\rm{i}}\infty }^{c + {\rm{i}}\infty } {\frac{{{x^s}}}{{s{{\left( {s - 1} \right)}^K}}}ds}  \notag\\
%&= \sum\limits_{a = 0}^1 {{\rm{Res}}\left\{ {\frac{{{x^s}}}{{s{{\left( {s - 1} \right)}^K}}},s = a} \right\}}
% \notag\\
&={\left( { - 1} \right)^K} + x\sum\limits_{k = 0}^{K - 1} {{{\left( { - 1} \right)}^k}\frac{{{{\left( {\ln{x}} \right)}^{K - k - 1}}}}{{\left( {K - k - 1} \right)!}}}.
\end{align}
where ${\rm Res}(f,a_k)$ denotes the residue of $f$ at $a_k$.
Substituting (\ref{eqn:der_hat_G_CDF_asy_fur}) into (\ref{eqn:out_prob_def_hat}) along with (\ref{eqn:gamma_snr_rel}), the asymptotic outage probability ${p_{out\_asy,K}}$ can be derived, as shown in the following theorem.
\begin{theorem}\label{the:pout_asy_sim}
The asymptotic outage probability ${p_{out\_asy,K}}$ can be decoupled as
\begin{equation}\label{eqn:out_pro_hat_G_asy_def}
{p_{out\_asy,K}} = \underbrace {{\mathcal G_{K}}({2^{\cal R}})}_A\underbrace {{{\left( {\ell \left( {\rho ,K} \right)} \right)}^{ - 1}}}_B\underbrace {\frac{1}{{{P_\Pi }\prod\limits_{k = 1}^K {{\sigma _k}^2} }} }_C, \, |\rho| \ne 1.
\end{equation}
where ${P_\Pi } = \prod\nolimits_{k = 1}^K {{P_K}} $. The impacts of transmission rate, channel time correlation and transmit powers on outage probability are clearly quantified through the terms of $A$, $B$ and $C$, respectively.
\end{theorem}
It is noteworthy that the asymptotic outage probability of HARQ-IR has never been accurately derived even under independent fading channels, which justifies the significance of our work.
\subsection{Discussions}
%According to Theorem \ref{the:pout_asy_sim}, the impacts of transmission rate and time correlation, and diversity order is specifically discussed as follows.
%\subsubsection{Diversity Order}
%The diversity order $d$ is defined as \cite{zheng2003diversity}
%\begin{equation}\label{eqn:diver_order_def}
%d =  - \mathop {\lim }\limits_{{\gamma} \to \infty } \frac{{\ln \left( {{{p}_{out,K}}} \right)}}{{\ln \left( {{\gamma}} \right)}}.
%\end{equation}
%By using (\ref{eqn:out_pro_hat_G_asy_def}), it follows that
%\begin{align}\label{eqn:diversity_order_def_der}
%d &= K, \quad |\rho| \ne 1.
%\end{align}
%Thus the diversity order of HARQ-IR is equal to the number of transmissions $K$, i.e., full diversity can be achieved by HARQ-IR even under time-correlated fading channels when $|\rho| \ne 1$. It is noteworthy that the same result does not hold under fully correlated fading channels, i.e., $|\rho|=1$. Under fully correlated fading channels, the diversity order reduces to $1$ since no time diversity can be achieved from retransmissions.
\subsubsection{Impact of Transmission Power}
It is shown in (\ref{eqn:out_pro_hat_G_asy_def}) that the asymptotic outage probability is influenced by the product of transmission powers in all HARQ rounds, i.e., ${P_\Pi }$. Therefore, the optimal power allocation for HARQ-IR over time-correlated fading channels can be enabled by adopting the same method developed in \cite{chaitanya2016energy}. For more details about optimal power allocation of HARQ, please consult with \cite{chaitanya2016energy}.
\subsubsection{Effect of Transmission Rate}
%Clearly, ${{g_{\bf{0}}}({2^{\cal R}})}$ should be studied to investigate the impact of transmission rate. From (\ref{eqn:sec_def_g_fun}), it is not difficult to prove that ${{g_{\bf{0}}}({2^{\cal R}})}$ is an increasing function of $\cal R$. Thus the increase of transmission rate leads to the decrease of outage probability.
%It is very difficult to evaluate ${g_{{\bf{0}}}}\left( {2^{\cal R}} \right)$ using (\ref{eqn:sec_def_g_fun}) due to its high computational complexity of tackling multiple-fold integral. To avoid that, an alternative simple representation of ${g_{{\bf{0}}}({2^{\cal R}})}$ is derived by simplifying the integral of (\ref{eqn:sec_def_g_fun}) as
%\begin{equation}\label{eqn:g_l_def}
%{g_{\bf 0}}\left(2^{\cal R} \right) = {\left( { - 1} \right)^K} + {2^{\cal R}}\sum\limits_{k = 0}^{K - 1} {{{\left( { - 1} \right)}^k}\frac{{{{\left( {\mathcal R\ln 2} \right)}^{K - k - 1}}}}{{\left( {K - k - 1} \right)!}}}.
%\end{equation}
Clearly from (\ref{eqn:out_pro_hat_G_asy_def}), the impact of transmission rate on asymptotic outage probability is determined by ${{\cal G}_K}\left( 2^{\cal R} \right)$. With (\ref{eqn:def_calG}), it can be proved that the first and the second derivative of ${{\cal G}_K}\left( 2^{\cal R} \right)$ with respect to $\cal R$ are non-negative, more precisely
\begin{align}\label{eqn:g_K_fir_der}
\frac{{d{\mathcal G_K}\left( {{2^{\cal R}}} \right)}}{{d{\cal R}}} &= \frac{{\ln \left( 2 \right)}}{{2\pi {\rm{i}}}}\int\nolimits_{{c} - {\rm{i}}\infty }^{{c} + {\rm{i}}\infty } {\frac{{{2^{{\cal R}s}}}}{{{{\left( {s - 1} \right)}^K}}}ds} \notag\\
%&= \ln \left( 2 \right)\sum\limits_{a = 1} {{\rm{Res}}\left\{ {\frac{{s{2^{{\cal R}s}}}}{{{{\left( {s - 1} \right)}^K}}},s = a} \right\}} \notag\\
 %&= \frac{{\ln \left( 2 \right)}}{{\left( {K - 1} \right)!}}{\left. {\frac{{{\partial ^{K - 1}}\left( {{2^{{\cal R}s}}} \right)}}{{\partial {s^{K - 1}}}}} \right|_{s = 1}} \notag\\
&= \frac{{\ln \left( 2 \right){{\left( {{\cal R}\ln \left( 2 \right)} \right)}^{K - 1}}}}{{\left( {K - 1} \right)!}}{2^{\cal R}} \ge 0,
\end{align}
\begin{multline}\label{eqn:g_K_fun_sec_der}
\frac{{d^2{\mathcal G_K}\left( {{2^{\cal R}}} \right)}}{{d{\cal R}^2}} = \frac{\left( {\ln 2} \right)^2}{2 \pi \rm i}\int\nolimits_{{c} - {\rm{i}}\infty }^{{c} + {\rm{i}}\infty } {\frac{{s{2^{{\cal R}s}}}}{{{{\left( {s - 1} \right)}^K}}}ds} \\
 = {\left( {\ln 2} \right)^2}\sum\limits_{a = 1} {{\mathop{\rm Res}\nolimits} \left\{ {\frac{{s{2^{{\cal R}s}}}}{{{{\left( {s - 1} \right)}^K}}},s = a} \right\}}
 \\
 %= \frac{{{{\left( {\ln 2} \right)}^2}}}{{\left( {K - 1} \right)!}}{\left. {\frac{{{\partial ^{K - 1}}\left( {s{2^{{\cal R}s}}} \right)}}{{\partial {s^{K - 1}}}}} \right|_{s = 1}} \\
 %= \frac{{{{\left( {\ln 2} \right)}^2}}}{{\left( {K - 1} \right)!}}{\left( {\sum\limits_{k = 0}^{K - 1} {C_{K - 1}^k{s^{\left( k \right)}}{{\left( {{2^{{\cal R}s}}} \right)}^{\left( {K - k - 1} \right)}}} } \right)_{s = 1}} \\
 %= \frac{{{{\left( {\ln 2} \right)}^2}}}{{\left( {K - 1} \right)!}}{\left( {{2^{{\cal R}s}}s{{\left( {{\cal R}\ln 2} \right)}^{K - 1}} + \left( {K - 1} \right){2^{{\cal R}s}}{{\left( {{\cal R}\ln 2} \right)}^{K - 2}}} \right)_{s = 1}} \\
 = \frac{{{{\left( {\ln 2} \right)}^2}}}{{\left( {K - 1} \right)!}} \left( {{2^{\cal R}}{{\left( {{\cal R}\ln 2} \right)}^{K - 1}} + \left( {K - 1} \right){2^{\cal R}}{{\left( {{\cal R}\ln 2} \right)}^{K - 2}}} \right) \\
 \ge 0.
\end{multline}
Thus ${\mathcal G_{K}}\left(2^{\cal R} \right)$ is an increasing and convex function of $\cal R$, which will greatly facilitate the optimal rate selection with numerous sophisticated convex methods.
%\subsubsection{Impact of Time Correlation}
%In (\ref{eqn:out_pro_hat_G_asy_def}), the term $\ell \left( {\rho,K} \right)$ quantifies the effect of time correlation. It can be proved from (\ref{eqn:ell_def}) that $\ell \left( {\rho,K} \right)$ is a decreasing function of $\rho$. Thus it follows that $\ell \left( {\rho,K} \right) \le \ell \left( {0,K} \right)=1$. Therefore, time correlation has a negative effect on outage probability, which consequently leads to the degradation of system performance.
\section{Numerical Results} \label{sec:num_res}
In this section, analytical results are verified. For illustration, we take systems with $\sigma_1=\cdots=\sigma_K=1$, $P_1=\cdots=P_K=P_T$, $\delta =1 $ and $R=2~\rm bps/Hz$ as examples. Notice that the exact outage probability in (\ref{eqn:out_prob_def_hat}) is calculated approximately by truncating the infinite series into a finite series with the constraint of $\sum\nolimits_{k = 1}^K {{\ell _k}} \le N$. In the following numerical analysis, the truncation order is set as $N=3$.

%To verify the effectiveness of truncation approach for the computation of $p_{out,K}$, the approximated outage probability $\tilde p_{out,K}$ is plotted against the truncation order $N$ with $K=4$ in Fig. \ref{fig:trun}. It is readily seen that $\tilde p_{out,K}$ with $N=3$ is enough to achieve a good approximation of $p_{out,K}$.
%To show the effect of truncation order $N$, the approximated outage probability after truncation $\tilde{\mathcal P}_{out}(K)$ is plotted versus $N$ with $K=4$ in Fig. \ref{fig:out_asy}. It is readily found that truncation order of $N=5$ is enough to well approximate ${\mathcal P}_{out}(K)$ with negligible error. In addition, low truncation order $N$ is sufficient to achieve a good approximation of ${\mathcal P}_{out}(K)$ under high SNR regime or low $\rho$. For example, the truncation order of $N=2$ can achieve a good approximation when $\rho=0.5$ or $P_T=10\rm dB$.
%\begin{figure}
%  \centering
%  % Requires \usepackage{graphicx}
%  \includegraphics[width=3in]{./fig/truncat.eps}\\
%  \caption{Effect of truncation order $N$.}\label{fig:trun}
%\end{figure}

%Thus the default truncation order is set as $N=3$ in the rest of paper.
In Fig. \ref{fig:div}, the exact and asymptotic outage probabilities are plotted against transmit power $P_T$ by setting $K=4$. Clearly, the exact results perfectly match with the simulation results and the asymptotic results coincide well with the exact/simulation results under high transmit power $P_T$, which justify the correctness of our analysis. For comparison, the outage probabilities obtained by using polynomial fitting technique \cite{shi2015analysis} are also presented in Fig. \ref{fig:div}. It can be seen that polynomial fitting technique does not provide a good approximation under high SNR regime due to the inability to guarantee the pointwise convergence. Thus this comparison further highlights the significance of the asymptotic outage analysis. In addition, note that diversity order quantifies the slope of outage probability against transmit power in a log-log scale. It is observed from Fig. \ref{fig:div} that curves under different time correlation become parallel as $P_T$ increases, which is consistent with our conclusion that time correlation does not influence diversity order.
\begin{figure}
  \centering
  \includegraphics[width=2.3in]{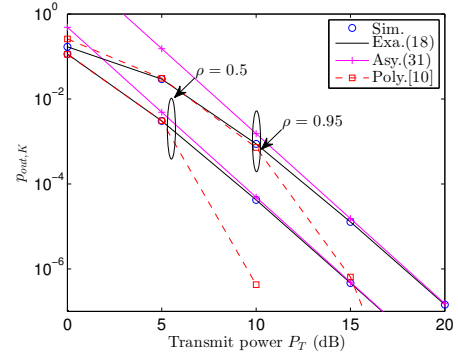}\\
  \caption{Outage probability versus transmit power $P_T$.}\label{fig:div}
\end{figure}

As shown in Fig. \ref{fig:coding}, $\mathcal G_K(2^{\cal R})$ is plotted against transmission rate $\cal R$ under different maximum numbers of transmissions $K$. As expected, $\mathcal G_K(2^{\cal R})$ is an increasing and convex function of $\cal R$, which justifies the correctness of our analysis. %Furthermore, it is also found that $\mathcal G_K(2^{\cal R})$ decreases with the increase of $K$.
\begin{figure}
  \centering
  \includegraphics[width=2.3in]{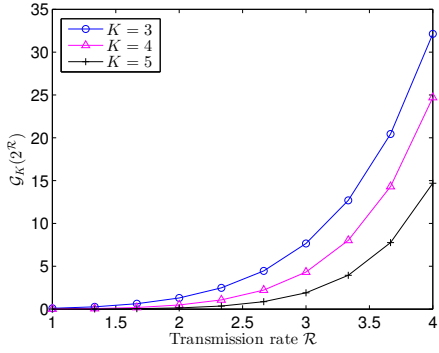}\\
  \caption{Impact of transmission rate.}\label{fig:coding}
\end{figure}

Fig. \ref{fig:corr} illustrates the effect of time correlation on the outage probability of HARQ-IR by setting $P_T=10$dB. It is readily found that the simulation results agree with the analytical results very well. In addition, it is verified that time correlation adversely affects the outage performance.
\begin{figure}
  \centering
  \includegraphics[width=2.3in]{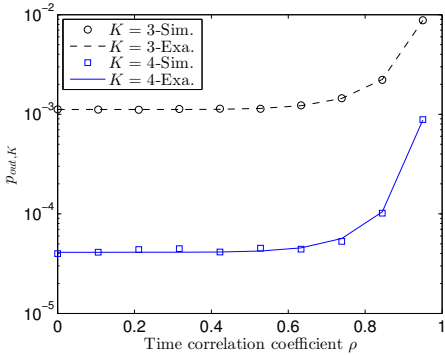}\\
  \caption{Impact of time correlation.}\label{fig:corr}
\end{figure}

\section{Conclusions} \label{sec:con}
In this paper, exact and asymptotic outage probabilities of HARQ-IR over time-correlated Rayleigh fading channels have been analyzed by developing a general analytical approach. In particular, the asymptotic outage probability has been derived in a simple form with which the impacts of time correlation, transmission rate and transmit powers are clearly quantified. %It has been revealed that low time correlation is beneficial to the outage performance of HARQ-IR and full time diversity can be achieved even under time-correlated fading channels.
The special form of asymptotic outage probability also eases the optimal system design, e.g., optimal power allocation and optimal rate selection.

\section{Acknowledgements}
This work was supported in part by National Natural Science Foundation of China under grants 61601524 and 61671488, in part by the Special Fund for Science and Technology Development in Guangdong Province under Grant No. 2016A050503025, in part by the Research Committee of University of Macau under grants MYRG2014-00146-FST and MYRG2016-00146-FST, and in part by the Macau Science and Technology Development Fund under grants 091/2015/A3 and 020/2015/AMJ.
\bibliographystyle{ieeetran}
\bibliography{closed_form}

\end{document}